\documentclass[submission,copyright,creativecommons]{eptcs}
\providecommand{\event}{PxTP 2021} 

\usepackage[utf8]{inputenc}
\usepackage[T1]{fontenc}
\usepackage{amssymb}
\usepackage{amsthm}
\usepackage{mathtools}
\usepackage{extarrows}
\usepackage{xspace}
\usepackage{wrapfig}
\usepackage[title]{appendix}
\usepackage{caption}
\usepackage{mathpartir}
\usepackage{siunitx}
\usepackage{cite}

\newcommand{\defeq}{\vcentcolon=}

\newcommand{\type}[3]{#1 \Vdash #2 : #3}

\newcommand{\bs}[2]{#1 \downarrow #2}

\newcommand{\txttt}[1]{\textsf{\upshape #1}}
\newcommand{\Split}{\txttt{KSplit}}
\newcommand{\Hole}{\txttt{KHole}}

\newcommand{\Assert}{\txttt{KAssert}}
\newcommand{\Trivial}{\txttt{KTrivial}}
\newcommand{\EqRefl}{\txttt{KEqRefl}}
\newcommand{\Destruct}{\txttt{KDestruct}}

\newcommand{\IntroQuant}{\txttt{KIntroQuant}}
\newcommand{\InstQuant}{\txttt{KInstQuant}}
\newcommand{\IntroType}{\txttt{KIntroType}}
\newcommand{\InstType}{\txttt{KInstType}}

\newcommand{\Rewrite}{\txttt{KRewrite}}
\newcommand{\Induction}{\txttt{KInduction}}

\newcommand{\bool}{\textsl{bool}}

\newcommand{\cert}{\textsl{cert}}
\newcommand{\pterm}{\textsl{term}_{poly}}
\newcommand{\mterm}{\textsl{term}_{mono}}
\newcommand{\ty}{\textsl{type}}
\newcommand{\ident}{\textsl{ident}}

\newtheorem{definition}{Definition}
\newtheorem{theorem}[definition]{Theorem}
\newtheorem{prop}[definition]{Proposition}
\newtheorem{example}[definition]{Example}

\usepackage{tikz}
\usetikzlibrary{arrows, shapes.misc, positioning}
\usetikzlibrary{matrix,calc,shapes,arrows,shadows,patterns,shapes.multipart}
\tikzset{>=stealth'}

\definecolor{thegray}{rgb}{0.9,0.9,0.9}
\definecolor{colorspec}{rgb}{0,0,0.797}
\definecolor{thered}{rgb}{0.797,0,0}
\definecolor{darkgreen}{rgb}{0.797,0,0}
\definecolor{theblue}{rgb}{0,0,0.797}
\definecolor{darkgray}{rgb}{0.8477,0.8477,0.8477}
\definecolor{ocaml-bg}{rgb}{0.9,0.9,0.9}
\definecolor{thegraygray}{rgb}{0.5,0.5,0.5}

\tikzstyle{module} = [rectangle, draw, top color=gray!10, bottom color=gray!10,
    text centered, rounded corners]
\tikzstyle{line} = [line width=1.5pt, >=latex, ->, color=thered, font=\sffamily,
    shorten <=2pt, shorten >=2pt]

\title{A Framework for Proof-carrying Logical Transformations}

\author{Quentin Garchery
\institute{Université Paris-Saclay, CNRS, Inria, LMF, 91405, Orsay, France}
}
\def\titlerunning{Proof-carrying Logical Transformations}
\def\authorrunning{Quentin Garchery}

\begin{document}
\maketitle

\begin{abstract}
  In various provers and deductive verification tools, logical transformations
  are used extensively in order to reduce a proof task into a number of simpler
  tasks. Logical transformations are often part of the trusted base of such
  tools. In this paper, we develop a framework to improve confidence in their
  results. We follow a modular and skeptical approach: transformations are
  instrumented independently of each other and produce certificates that
  are checked by a third-party tool. Logical transformations are considered in a
  higher-order logic, with type polymorphism and built-in theories such as
  equality and integer arithmetic. We develop a language of proof certificates
  for them and use it to implement the full chain of certificate generation and
  certificate verification.
\end{abstract}


\section{Introduction}
\label{sec:introduction}

\paragraph{General Context and Motivation.}
Verifying a program is meant to improve its soundness guarantees
and relies on the trust towards the verification tool. Given how difficult it
can be to verify relatively simple programs, most tools try to simplify this
process and to make it as automatized as possible, which can drastically extend
their trusted code base.

Consider deductive program verification, where the program to verify is
annotated, and, in particular, given a specification. In this setting, the code
is analyzed against its specification thus generating \emph{proof tasks},
logical statements upon which depends the program correctness. To discharge a
proof task, one can first apply a \emph{logical transformation} which reduces it
to a number of new proof tasks which are hopefully easier to discharge.
Transformations are powerful tools that can, for example, be applied to
translate a task into a prover's logic before calling~it.

The main objective of this article is to improve trust in those logical
transformations. This work is decisive because logical transformations are
general and can be used in many different settings. We apply our method to the
deductive program verification tool Why3~\cite{bobot14sttt}, which makes
extensive use of logical transformations. In fact, they are at the core of its
interactive theorem proving feature and are necessary to be able to encode proof
tasks into the logic of one of the dozens of third-party automatic theorem
provers available inside Why3. The implementation of the transformations adds up
to a total of more than \num{17000} lines of OCaml code. This code being in the
trusted Why3 code base, it represents an interesting case study.

\begin{example}\label{intro_example}
Suppose that we have a proof task where we have to prove that $p~ (y * y)$ holds
for any integer $y$ of the form $y = 2 * x + 1$ and any integer predicate $p$
that satisfies the hypothesis $H$ stating that $\forall i : int. ~ p~ (4 * i +
1)$.
Finding how to instantiate the hypothesis $H$ is difficult or even impossible for
some automatic theorem provers so we cannot automatically discharge this
task. The transformation \texttt{instantiate}, defined in Why3, simply
instantiates an hypothesis and when called with arguments $H$ and $x*x+x$ on the
given task, produces the same task but with the added hypothesis that states
that $p~(4*(x*x+x)+1)$. Provers won't have to instantiate the hypothesis $H$ to
discharge this new task. In fact, this task can now be discharged by theorem
provers capable of handling arithmetic goals on condition of translating it into
the logic of the prover in question. This translation is also being done with the
help of transformations.
The entire process described in this example is to be trusted to ensure
correctness of the initial program.
\end{example}

\paragraph{Contributions and Overview.}
In this article, we describe a practical framework to validate logical
transformations. In order to define what it means for a transformation to be
correct, the logical setting of proof tasks is detailed in
Section~\ref{setting}. We follow a skeptical
approach~\cite{barendregt2002autarkic}: certificates, defined in
Section~\ref{certificates}, are generated every time a transformation is applied
and are checked independently at a later time. Contrary to the autarkic
approach, used for example for some automatic theorem provers~\cite{Lescuyer11},
which would consist here in verifying directly the transformations, the
skeptical approach has the benefit of not fixing the implementation of the
transformations. Our work is based on certificates with holes, a notion that is,
to our knowledge, new in the setting of the skeptical approach. This allows for
modular development, where certificates can be built incrementally and
transformations can be composed and defined independently. We extend our
framework with some key interpreted theories in Section~\ref{interp} and show
how to do so for any other interpreted theory along the way. The checkers for
our certificates can also be defined independently, as it is done in
Section~\ref{checkers}. In fact, we designed two checkers and one of them is
based on Lambdapi/Dedukti~\cite{lambdapi16}, an off-the-shelf proof
assistant. This has also led us to develop a translation procedure for proof
tasks into the $\lambda\Pi$-Calculus modulo rewriting. This approach, while
applicable to logical transformations in general, has been applied to the
program verification tool Why3 for a number of its transformations including
transformations dealing specifically with higher-order logic. We conclude this
article by evaluating this application to Why3 in Section~\ref{experiments}. The
source code for the whole work described in this article is available in the
Why3 repository~\cite{branchpxtp}.

\section{Logical Setting} \label{setting}

We present the logical setting used throughout this article. The goal here is
to define logical transformations and the proof tasks they are applied to.

\subsection{Types and Terms} \label{terms}

Proof tasks are formed from typed terms and those terms are meant to designate
both the terms from the program and the formulas stating properties about
them. We use the Hindley-Milner type system~\cite{milner1978theory} except that
our terms are explicitly quantified over types. Names are taken from an
infinite set of available identifiers which is designated by $\ident$.

Types are described by a \emph{type signature} $I$, a set of pairs of the form
`$\iota : n$' composed of an $\ident$ called type symbol and an integer
representing its arity. Sets are denoted by separating their elements with
commas. Note that according to the following grammar, type symbols are always
completely applied.
\[
\begin{array}{cclcr}
\ty & ::= & \alpha      &\quad\quad& \text { type variable} \\
   & |   & prop                  && \text { type of formulas}\\
   & |   & \ty \leadsto \ty        && \text { arrow type} \\
   & |   & \iota (\ty, \dots, \ty) && \text { type symbol application}\\
\end{array}
\]

Terms have polymorphic types and new terms can be built by quantifying over
terms of any type. Quantification over type variables is explicit and restricted
to only be in the prenex form.  Note that the application uses the Curry
notation, i.e., a function term is applied to a single argument term at a
time. The application is left-associative and the type arrow $\leadsto$ is
right-associative.

 \[
\begin{array}{cclcr}
\pterm &  \vcentcolon \vcentcolon = & \mterm  &\quad\quad&  \\
   & |   & \Pi \alpha. ~ \pterm                 && \text {type quantifier} \\
   &     & \\
\mterm &  \vcentcolon \vcentcolon = & x            && \text {variable} \\
   & |   & \top                                   && \text {true formula}\\
   & |   & \bot                                   && \text {false formula}\\
   & |   & \lnot ~ \mterm                          && \text {negation}\\
   & |   & \mterm ~ op ~ \mterm                      && \text {logical binary operator}\\
   & |   & \mterm ~ \mterm                          && \text {application} \\
   & |   & \lambda x : type. ~ \mterm                && \text {anonymous function} \\
   & |   & \exists x : type. ~ \mterm                && \text {existential quantifier} \\
   & |   & \forall x : type. ~ \mterm                && \text {universal quantifier} \\

   &     & \\
op &  \vcentcolon \vcentcolon= & \land                     && \text {conjonction} \\
   & |    & \lor                      && \text {disjunction} \\
   & |    & \Rightarrow           && \text {implication} \\
   & |    & \Leftrightarrow       && \text {equivalence} \\
\end{array}
\]

The \emph{(term) substitution} of variable $x$ by term $u$ in term $t$ is
written $t[x \mapsto u]$ and $t[\alpha \mapsto \tau]$ is the \emph{(type)
substitution} of type variable $\alpha$ by type $\tau$ in term $t$. A {\em
  signature} $\Sigma$ is a set of pairs of the form `$x : \tau$' composed of a
variable and its type; this type should be understood as quantified over all of
its type variables.

\begin{definition}[Typing]
  We write ${\type{I ~|~ \Sigma}{t}{\tau}}$ when the term~$t$ has no free type
  variables and has type~$\tau$ in type signature~$I$ and signature~$\Sigma$. We
  omit $I$ when it is clear from the context.
\end{definition}

The complete set of rules defining this predicate is given in
Appendix~\ref{trules}. Remark that, in the case where $t$ is an element of
$\mterm$ then the predicate ${\type{\Sigma}{t}{\tau}}$ implies that $t$ has no
type variables: $t$ is monomorphic.

\subsection{Proof Tasks} \label{tasks}

Proof tasks represent sequents in higher-order logic, they are formed from two
sets of premises: a set of hypotheses and a set of goals. A
\emph{premise} is a pair of the form `$P : t$' composed of an $\ident$ and
a $\pterm$ representing a formula.

\begin{definition}[Proof Task]
  Let $I$ be a type signature, $\Sigma$ be a signature, $\Gamma$ and $\Delta$ be
  sets of premises. {\em Proof tasks} are denoted by
  ${I ~|~ \Sigma ~|~ \Gamma \vdash \Delta}$ which represents the sequent where
  goals, given by $\Delta$, and hypotheses, given by $\Gamma$, are written in
  the signature~$\Sigma$ with types in $I$. We allow ourselves to omit $I$ and,
  possibly, $\Sigma$, when they are clear from the context.
\end{definition}

A task ${T \defeq I ~|~ \Sigma ~|~ \Gamma \vdash \Delta}$ is said to be
\emph{well-typed} when every premise $P : t$ from $\Gamma$ or $\Delta$ is
such that $\type{I ~|~ \Sigma}{t}{prop}$. The \emph{validity} of a task is only
defined when it is well-typed. In this case, the task $T$ is said to be valid
when every model of~$I$, $\Sigma$ and every formula in $\Gamma$ is also a model
of some formula in~$\Delta$.

\begin{example}
  Consider the task ${I ~|~ \Sigma ~|~ \Gamma \vdash \Delta}$ with
  \begin{align*}
    I \defeq ~& color : 0, ~set : 1 \\
    \Sigma \defeq~ &red : color(), ~green : color(), ~blue : color(), \\
             &empty : set (\alpha),
    ~add : \alpha \leadsto set(\alpha) \leadsto set (\alpha),\\
             &mem : \alpha \leadsto set (\alpha) \leadsto prop\\
    \Gamma \defeq ~& H_1 : \Pi \alpha. ~\forall x : \alpha. ~\forall y : \alpha.~
                    \forall s : set (\alpha). ~mem ~x ~s \Rightarrow
                    mem ~x ~(add ~y ~s),\\
             &H_2 : \Pi \alpha. ~\forall x : \alpha. ~\forall s : set (\alpha).~
               mem ~x ~(add ~x ~s)\\
    \Delta \defeq ~& G : mem ~green ~ (add ~red ~(add ~ green ~empty))
  \end{align*}
  This task defines the types $color$ and $set$ with associated symbols $red$,
  $green$, $blue$, $empty$, $add$ and $mem$. The type symbol declaration $set :
  1$ defines a type symbol $set$ of arity~$1$ for polymorphic sets. The
  signature declaration $add : \alpha \leadsto set (\alpha) \leadsto set
  (\alpha)$ allows us to declare a function that can be used to add an element
  of any type to a set containing elements of the same type. This task also
  defines hypotheses such that the predicate $mem$ holds if its first argument
  is contained in the second argument. For instance, the hypothesis~$H_2$ is
  applicable to sets of any type, and states that every set contains the element
  that has just been added to it. With the given goal, this task is valid.
\end{example}

\subsection{Logical Transformations} \label{transfos}

A \emph{logical transformation} is a function that takes a task as input and
returns a list of tasks. Lists are denoted by separating their elements with
semicolons. We say that a transformation is applied on an \emph{initial task}
and returns \emph{resulting tasks}.
A transformation can fail, in this case the whole process is terminated and
we do not have to prove the correctness of the application.
If a transformation succeeds, we want the verification to be based on the
validity of the resulting tasks and to be able to forget about the initial
task. This is why we say that a \emph{transformation application is correct} when the
validity of each resulting task implies the validity of the initial task.
In Example~\ref{intro_example}, the transformation \texttt{instantiate} returns
the initial task modified by adding the instantiated hypothesis to it.  When a
transformation application is correct, it only remains to prove that this
resulting task is valid in order to make sure that the initial task is also
valid. This is our approach: we certify transformation applications, thus
relating initial and resulting tasks.

\section{Certificates} \label{certificates}

To verify a transformation, we instrument it to produce a certificate and check
each application of the transformation thanks to the corresponding
certificate. We first define our own certificate format with the goal of making
the verification of those certificates as easy as possible. We show how to
improve modularity and ease of use in a second time in Section~\ref{elab}.

\subsection{Syntax}

An excerpt of the recursive definition of certificates is given in
Figure~\ref{fig:cert}. More certificates will be detailed on their own in
Section~\ref{interp} and the others won't be presented in this article for
brevity. We call these certificates the \emph{kernel certificates}. To make
certificates easier to check, we design them in such a way that they are very
precise. For example, the Boolean values indicate whether the premise is an
hypothesis or a goal. Moreover, the kernel certificates have voluntarily been
kept as elementary as possible and this makes it easier to trust them. In
particular, this approach makes it easier to check every case (about 20 of them)
when proving by induction a property of correctness of kernel certificates, as
it is done in paragraph~\ref{proofterm}.

The certificates can contain tasks and each $\Hole$ node carries one of those
tasks. When $c$ is a certificate, \emph{the leaves of}~$c$ designate the list of
all tasks obtained by collecting them (in the $\Hole$ nodes) when doing an
in-order traversal of the certificate tree. The leaves of a certificate are
meant to be, in the end, the resulting tasks of the transformation it is
certifying. A certificate with holes associated to a transformation application
can be checked without needing to wait for the proof of the returned tasks to
fill its holes, and this is what makes our certificates original. This design
choice has been guided by our will for modularity: we want to progressively
certify logical transformations.

\begin{figure}[tp]
  \centering
  \[
  \begin{array}{rcl}
    \cert & \vcentcolon \vcentcolon = &
    \Hole(\textsl{task})\\
    & \mid &
    \Trivial(\bool, \ident) \\
    & \mid &
    \Assert(\ident,\pterm,\cert,\cert)\\
    & \mid &
    \Split(\bool, \mterm, \mterm, \ident,\cert,\cert)\\
    & \mid &
    \Destruct(\bool, \mterm, \mterm, \ident,\ident,\ident,\cert)\\
    & \mid &
    \IntroQuant(\bool, \ty, \mterm, \ident, \ident,\cert)\\
    & \mid &
    \InstQuant(\bool, \ty, \mterm, \ident,\ident,\mterm,\cert) \\
    & \mid &
    \IntroType(\pterm, \ident, \ident, \cert)  \\
    & \mid &
    \InstType(\pterm, \ident,\ident,\ty,\cert) \\
    & \mid &
    \dots \\

  \end{array}
  \]
  \caption{Definition of Kernel Certificates (excerpt)}
  \label{fig:cert}
  \hrulefill
\end{figure}

\begin{figure}[t]
    \centering
\begin{mathpar}
  \inferrule{~}
  {\bs{\Gamma\vdash \Delta}
    {\Hole(\Gamma \vdash \Delta)}}

  \inferrule{~}
  {\bs{\Gamma, H:\bot \vdash \Delta}
    {\Trivial(false, H)}}

  \inferrule{~}
  {\bs{\Gamma \vdash \Delta, G:\top}
    {\Trivial(true, G)}}

  \inferrule{\bs{\Sigma ~|~ \Gamma \vdash \Delta, P : t}{c_1} \\
    \bs{\Sigma ~|~ \Gamma, P : t \vdash \Delta}{c_2}\\
    \type{\Sigma}{t}{prop}}
  {\bs{\Sigma ~|~ \Gamma \vdash \Delta}
    {\Assert(P,t,c_1,c_2)}}

  \inferrule{\bs{\Gamma, H : t_1 \vdash \Delta}{c_1}\\
    \bs{\Gamma, H : t_2 \vdash \Delta}{c_2}}
  {\bs{\Gamma, H : t_1 \lor t_2 \vdash \Delta}
                  {\Split(false, t_1, t_2, H,c_1,c_2)}}

  \inferrule{\bs{\Gamma \vdash \Delta, G : t_1}{c_1}\\
    \bs{\Gamma \vdash \Delta, G : t_2}{c_2}}
  {\bs{\Gamma \vdash \Delta, G : t_1 \land t_2}
                  {\Split(true, t_1, t_2, G,c_1,c_2)}
                  }

  \inferrule{\bs{\Gamma, H_1 : t_1, H_2 : t_2 \vdash \Delta}{c}}
  {\bs{\Gamma, H : t_1 \land t_2 \vdash \Delta}
                 {\Destruct(false, t_1, t_2, H,H_1,H_2,c)}
                 }

  \inferrule{\bs{\Gamma \vdash \Delta, G_1 : t_1, G_2 : t_2}{c}}
  {\bs{\Gamma \vdash \Delta, G : t_1 \lor t_2}
                 {\Destruct(true, t_1, t_2, G,G_1,G_2,c)}
                 }

  \inferrule{\bs{\Sigma, y : \tau ~|~ \Gamma, H : t [x \mapsto y] \vdash \Delta}{c}\\
    y \text{ is fresh w.r.t. } \Sigma, \Gamma, \Delta, t}
  {\bs{\Sigma ~|~ \Gamma, H : \exists x : \tau . ~t \vdash \Delta}
                  {\IntroQuant(false, \tau, \lambda x : \tau. ~t,H,y,c)}
                  }

  \inferrule{\bs{\Sigma, y : \tau ~|~ \Gamma\vdash \Delta, G : t [x \mapsto y]}{c}\\
    y \text{ is fresh w.r.t. } \Sigma, \Gamma, \Delta, t}
  {\bs{\Sigma ~|~ \Gamma\vdash \Delta, G : \forall x : \tau. ~t}
                  {\IntroQuant(true, \tau, \lambda x : \tau. ~t,G,y,c)}
                  }

  \inferrule{\bs{\Sigma ~|~ \Gamma, H_1 : \forall x : \tau. ~t, H_2 : t[x \mapsto u] \vdash \Delta}{c}\\
    \type{\Sigma}{u}{\tau}}
  {\bs{\Sigma ~|~ \Gamma, H_1 : \forall x : \tau. ~t \vdash \Delta}
                 {\InstQuant (false, \tau, \lambda x : \tau. ~t,H_1,H_2,u,c)}
                 }

  \inferrule{\bs{\Sigma ~|~ \Gamma \vdash \Delta, G_1 : \exists x : \tau. ~t, G_2 : t [x \mapsto u]}{c}\\
    \type{\Sigma}{u}{\tau}}
  {\bs{\Sigma ~|~ \Gamma \vdash \Delta, G_1 : \exists x : \tau. ~t}
                 {\InstQuant(true, \tau, \lambda x : \tau. ~t,G_1,G_2,u,c)}
                 }

  \inferrule{\bs{I, \iota : 0 ~|~ \Sigma ~|~ \Gamma\vdash \Delta, G : t [\alpha \mapsto \iota]}{c}\\
    \iota \not\in I}
  {\bs{I ~|~ \Sigma ~|~ \Gamma\vdash \Delta, G : \Pi \alpha.~ t}
                  {\IntroType(\Pi \alpha. ~t, G,\iota,c)}
                  }

  \inferrule{\bs{\Gamma, H_1 : \Pi \alpha. ~t, H_2 : t [\alpha \mapsto \tau] \vdash \Delta}{c}\\
    \tau \text{ has no type variables}}
  {\bs{\Gamma, H_1 : \Pi \alpha.~ t \vdash \Delta}
                 {\InstType (\Pi \alpha. ~t, H_1,H_2,\tau,c)}
                 }
\end{mathpar}
\caption{Certificate Rules (excerpt)}
\label{fig:crules}
\hrulefill
\end{figure}

\subsection{Semantics}\label{semcert}

The semantics of certificates is defined by a binary predicate $\bs{T}{c}$,
linking the initial task~$T$ to a certificate~$c$. Informally, the predicate
$\bs{T}{c}$ holds if $c$ represents a proof of the fact that the validity of the
leaves of $c$ implies the validity of $T$. In Figure~\ref{fig:crules}, we give
the rules that cover the certificates from Figure~\ref{fig:cert}. In this sense,
the rules are only an excerpt of the complete set of rules (given in
Appendix~\ref{crules}) defining the predicate $\bs{T}{c}$. Notice that
some of the certificates have dual rules for the hypotheses and the goals.

The certificate $\Hole$ is used to validate transformations that have resulting
tasks and can be used directly for the identity transformation. The certificate
$\Trivial$ is used to validate a transformation application that has no
resulting task when the initial task contains a trivial premise. The
$\Assert$ certificate allows to introduce a cut on a polymorphic
formula. Remember that the side condition implies that this formula cannot have
free type variables. The certificates $\Split$ and $\Destruct$ are used to
validate a transformation application that first splits a
premise~$H$. Certificate $\IntroQuant$ is used to introduce the variable of a
quantified premise. The certificate $\InstQuant$ is used to instantiate a
quantified premise with a term and the side condition ensures that this term is
monomorphic. The certificates $\IntroType$ and $\InstType$ are used to deal with
type-quantified premises.
\begin{example}
  Let~$T$ and $T_{inst}$ denote, respectively, the initial and the resulting
  task from Example~\ref{intro_example}, and let~$H_{inst}$ be the name of the
  new instantiated hypothesis. Suppose that symbol~$plus$, symbol~$mult$ and
  type symbol~$int$ are defined in the signature and type signature, then we
  have $\bs{T}{c}$ with
  \begin{align*}
    c \defeq \InstQuant (&false, int, \lambda i : int. ~p~ (plus~ (mult~ 4~ i) ~1), \\
    &H, H_{inst}, int, plus ~ (mult ~ x ~x) ~x, \Hole (T_{inst}))
  \end{align*}
\end{example}

\begin{prop}
  If $T$ is well-typed then every task in a derivation $\bs{T}{c}$ is also
  well-typed.
\end{prop}
We also assume that transformations are always applied to well-typed initial
tasks and produce well-typed resulting tasks (or they fail), so that every task
we consider from now on is implicitly assumed to be well-typed.

\begin{theorem}[Certificate Correctness] \label{ccorrect} If $\bs{T}{c}$ then
  the validity of each leaf of $c$ implies the validity of $T$.
\end{theorem}

\begin{proof} By induction on $\bs{T}{c}$.
\end{proof}

\subsection{Design Choices}

The certificate rules are taken from the sequent calculus LK rules with
modifications for two reasons. First, we want the production of certificates to
be more natural. This is why the name $\Split$ is well-suited for a
transformation application that, from the initial task $\Gamma, H : t_1 \lor t_2
\vdash \Delta$, returns $\Gamma, H : t_1 \vdash \Delta$ and $\Gamma, H : t_2
\vdash \Delta$. Indeed, it would be confusing to say that, from this initial
task, the transformation does the left introduction of the disjunction. Second,
we want to be able to implement a checker of certificates following these
rules. To this end, instead of asking the checkers to find the names that were
chosen by the transformation, we register these names in the certificates. For
example, the $\IntroQuant$ certificate mentions the name $y$ of the new fresh
variable that is being introduced and this is reflected in the corresponding
rules.

\subsection{Certifying Transformations and Composition}
\begin{definition}[Certifying transformation]
  A \emph{certifying transformation} is a transformation that, applied on an
  initial task, produces, on top of a list~$L$ of resulting tasks, a
  certificate~$c$ such that $L$ is the leaves of~$c$. We say that we
  instrumented the transformation to produce a certificate.
\end{definition}

Composing transformations is useful to define a transformation from simpler
ones. To compose certifying transformations, one also needs to be able to
substitute certificates, that is, to replace a $\Hole$ in one certificate
with another certificate. This composition allows for a modular development of
certifying transformations.

\section{Adding Support for Interpreted Theories} \label{interp}
For now, our formalism implicitly makes the assumption that every symbol is
uninterpreted: they are taken as fresh new symbols for every task. Still, we
want some symbols (such as equality or arithmetic operations) to have a fixed
interpretation. Moreover, some transformations, like induction, use specific
theories and we need to add certificate steps to be able to certify them.

To make sure that the interpretation is unique, we should not quantify over the
interpreted symbols at the level of the tasks. Interpreted symbols are not part
of the signature or type signature of tasks. This ensures that the
interpretation stays the same for the initial task and for the resulting tasks
and this is enough to handle transformations on tasks that contain interpreted
symbols. To handle transformations that deal with the interpreted symbols and
use their properties, we extend our certificate format and add rules
corresponding to their properties.

\subsection{Polymorphic Equality}

The polymorphic equality is interpreted. To obtain the usual properties of the
equality, we add the certificates:
\begin{align*}
  &\EqRefl(\mterm, \ident) \\
  &\Rewrite(\bool, \mterm, \mterm, \mterm, \ident, \ident, \cert)
\end{align*}
and three kernel rules:
\begin{mathpar}
  \inferrule{~}
  {\bs{\Sigma ~|~ \Gamma \vdash \Delta, G: x = x}
    {\EqRefl(x, G)}}

  \inferrule{\bs{\Gamma, H :  a = b, P : t[b] \vdash \Delta}{c}}
  {\bs{\Gamma, H : a = b, P : t[a] \vdash \Delta}
    {\Rewrite(false, a, b, t, P, H,c)}}

  \inferrule{\bs{\Gamma, H :  a = b \vdash \Delta, P : t[b]}{c}}
  {\bs{\Gamma, H : a = b \vdash \Delta, P : t[a]}
    {\Rewrite(true, a, b, t, P,H,c)}}
\end{mathpar}
When $t$ is a function of the form $\lambda x. ~u$, we write $t[u']$ for the
substitution $u[x \mapsto u']$.  These rules deal with the reflexivity of
equality and the rewriting under context. They are sufficient to obtain the
standard properties of equality: symmetry, transitivity, and congruence.

\paragraph{Application to the \texttt{rewrite} Transformation.}

The Why3 \textsl{rewrite} transformation is a powerful transformation that can
rewrite terms modulo an equality that is under implications and universal
quantifiers. It looks for a substitution to match the left-hand side of the
given equality to rewrite it as the right-hand side following this
substitution. Moreover, it allows rewriting from right to left instead.
We instrument this transformation in the general case: using the found
substitution, we define certificates to introduce in turns implications and
universal quantifiers in a temporary hypothesis, to then apply symmetry of
equality if needed and to use this equality to rewrite the target premise
and finally remove the temporary hypothesis.

\subsection{Integers}

The type symbol~$int$, integer literals and the operator symbols~$+$, $*$, $-$,
$>$, $<$, $\geq$ and $\leq$ are interpreted.
To be able to certify a transformation that performs an induction on integers,
we add a certificate $\Induction(\ident, \mterm, \mterm, \ident, \ident, \ident,
\cert, \cert)$ to the kernel certificates with one rule for strong induction:

\begin{mathpar}
  \inferrule{i \text{ is fresh w.r.t. } \Gamma, \Delta, t \\
        \type{\Sigma}{i}{int}\\
        \type{\Sigma}{a}{int}\\
    \bs{\Gamma, H_i : i \leq a \vdash \Delta, G : t[i]}{c_{base}}\\
    \bs{\Gamma, H_i : i > a, H_{rec} : \forall n : int, ~ n < i \Rightarrow t[n] \vdash \Delta, G : t[i]}{c_{rec}}}
  {\bs{\Gamma \vdash \Delta, G : t[i]}{\Induction(i, a, t, G, H_i, H_{rec}, c_{base}, c_{rec})}}
\end{mathpar}

\paragraph{Application to the \texttt{induction} Transformation.}

The Why3 \texttt{induction} transformation can be called even if the context
depends on the integer on which the induction is done and, in this case, the
induction hypothesis takes into account this context.
This transformation is instrumented to produce the certificate that first puts
in the goal the premises that depend on the integer on which the induction is
done, applies the $\Induction$ certificate, and then, for the two resulting
tasks, introduces those premises.

\section{Certificate Checker} \label{checkers}

Let us consider an application of a certifying transformation~$\phi$ on an
initial task~$T$ that produces a certificate~$c$ and a list of resulting
tasks~$L$. To verify that this application is correct, $c$, $T$ and $L$ are
provided as input to some checker. If the checker validates this application
then the transformation returns $L$, otherwise it fails. In this section, we
present the elaboration of certificates, a preprocessing step realized before
calling any checker. We then show how to define such checkers and be confident
in their answers.

\subsection{Surface Certificates and Elaboration} \label{elab}

\begin{wrapfigure}{r}{0.56\textwidth}
  \hspace{10mm}\includegraphics[width=0.45\textwidth]{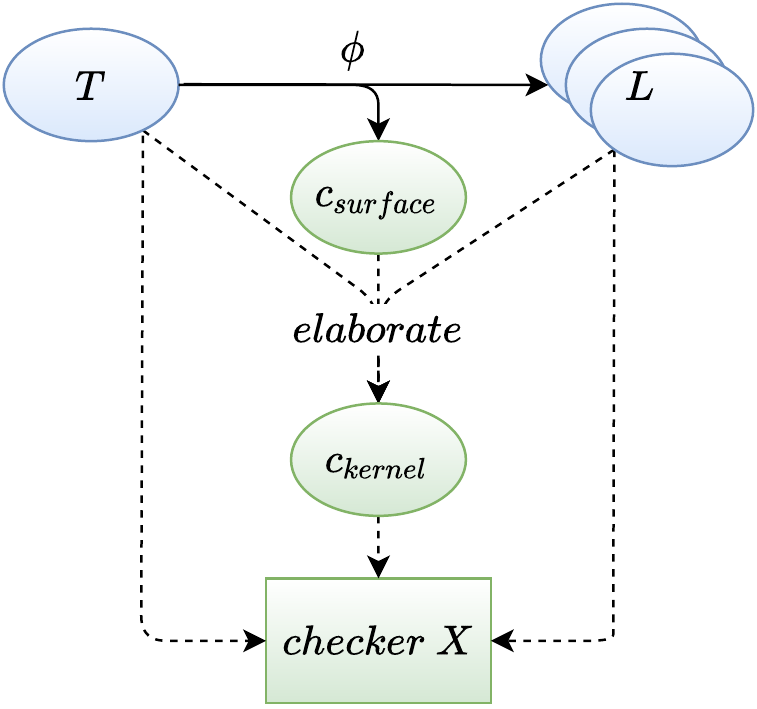}
\end{wrapfigure}
Making a transformation certifying can be difficult, especially if the
transformation has to produce a low level certificate. To facilitate this
process, we define \emph{surface certificates} that are easier to use than
kernel certificates. We instrument the transformations to produce surface
certificates instead of kernel certificates and implement an \emph{elaboration}
procedure to translate them into kernel certificates. In order to obtain the
needed data to produce the kernel certificate, before calling a checker, the
elaboration procedure is called with the initial task and the resulting tasks
as input.
Because the proof of correctness of certificates is done on the kernel
certificates, we can define more complex surface certificates, as long as we are
able to define the elaboration for them. Another advantage of surface
certificates over kernel certificates is that they are less verbose, making them
easier to produce.

\begin{example} \label{ex:split}
  The surface certificate $\txttt{SSplit}(\ident, \cert, \cert)$ is elaborated
  into the kernel certificate $\Split$. Suppose that a certifying transformation
  applied on initial task $T \defeq H : x_1 \lor x_2 \vdash G : x$ returns the
  list $T_1; ~ T_2$ with $T_1 \defeq H : x_1 \vdash G : x$ and $T_2 \defeq H :
  x_2 \vdash G : x$ and the surface certificate
  \[\txttt{SSplit} (H, \Hole(T_1), \Hole (T_2))\]
  The elaboration produces a kernel certificate indicating which formulas it is
  applied to ($x_1$ and $x_2$) and that $H$ is not a goal (Boolean $false$):
  \[\Split (false, x_1, x_2, H, \txttt{EHole}(T_1), \txttt{EHole}(T_2))\]
\end{example}

We can define every surface certificate that we find convenient. For now there
are about 10 more surface certificates than kernel certificates. Among them, there
are $\txttt{SEqSym}$ and $\txttt{SEqTrans}$ for symmetry and
transitivity of equality and $\txttt{SConstruct}$ described in the
following example.

\begin{example}
  We define the surface certificate
  $\txttt{SConstruct}(\ident,\ident,\ident,\cert)$ to validate a transformation
  application that first merges two premises into one. More precisely, by
  writing $c'$ a certificate $c$ that has been elaborated, we should be able to
  derive the following rules:
  \begin{mathpar}
    \inferrule{\bs{\Gamma, P : t_1 \land t_2 \vdash \Delta}{c'}}
              {\bs{\Gamma, P_1 : t_1, P_2 : t_2 \vdash \Delta}
                {\txttt{SConstruct}(P_1, P_2, P, c)'}
              }

    \inferrule{\bs{\Gamma \vdash \Delta, P : t_1 \lor t_2}{c'}}
              {\bs{\Gamma \vdash \Delta, P_1 : t_1, P_2 : t_2}
                {\txttt{SConstruct}(P_1, P_2, P, c)'}
              }
  \end{mathpar}

  The $\txttt{SConstruct}$ certificate does not have a corresponding kernel
  certificate. Instead, it is replaced during elaboration by a combination of
  the $\Assert$ certificate on formula $t_1 \land t_2$ and other propositional
  certificates, notably $\Destruct$. Notice that we need to have access to the
  formula $t_1 \land t_2$, which is precisely the point of the elaboration and
  why we could not define directly $\txttt{SConstruct}$ as a combination of
  surface certificates.
\end{example}

\subsection{OCaml Checker} \label{ochecker}

We implemented two checkers, the first one is written in OCaml and follows a
computational approach: it is based on a function $ccheck$ that is called with
the certificate~$c$ and initial task~$T$ and interprets the certificate as
instructions to derive tasks such that their validity implies the validity of~$T$, verifying in the end that the derived tasks are the leaves of~$c$. The
checker validates the application when this function returns~$true$.

\begin{theorem}
  Let $c$ be the certificate produced by applying a certifying transformation on
  task~$T$. If $ccheck ~c ~T = true$ then $\bs{T}{c}$.
\end{theorem}

The OCaml-checker definition follows closely the semantics of the certificates.
For this reason, the proof of the previous theorem is relatively straightforward
to do on paper. Together with Theorem~\ref{ccorrect}, this guarantees that each
application of a transformation that is checked by the OCaml checker is correct.

\subsection{Lambdapi Checker} \label{lchecker}

The second checker uses Lambdapi/Dedukti~\cite{lambdapi16}, a proof assistant
based on a type checker for the $\lambda\Pi$-Calculus modulo rewriting which
extends the $\lambda\Pi$-Calculus formalism with custom rewriting rules. This
checker has two main advantages over the OCaml checker: (1)~this checker uses an
off-the-shelf proof assistant, benefiting from its reliability and its features,
such as the ability to define custom rewriting rules; (2)~this checker is proven
to be correct, and this proof is machine-checked.

\begin{wrapfigure}{r}{0.50\textwidth}
\vspace{-3mm}
\includegraphics[width=0.45\textwidth]{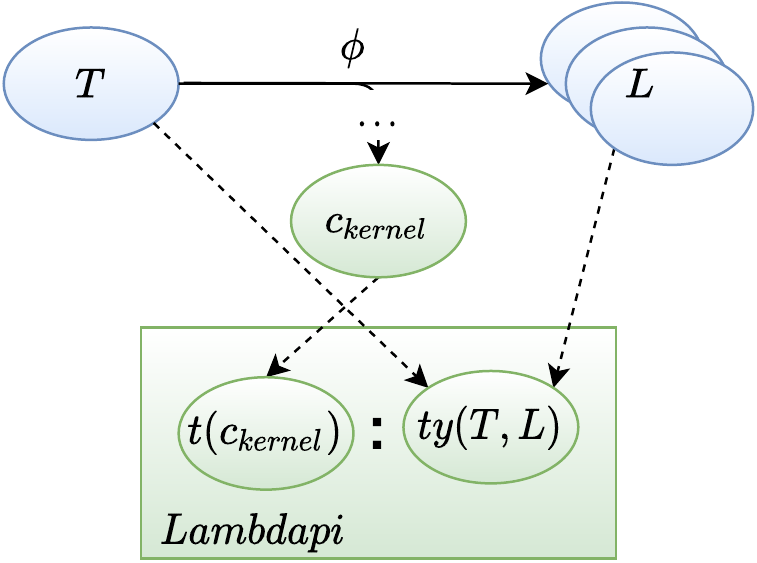}
\vspace{-4mm}
\end{wrapfigure}
Every time a transformation is called, a Lambdapi proof is generated, and this
proof is then checked by the type checker. More precisely, we define a shallow
embedding of proof tasks in Lambdapi: a proof task $T$ is encoded as a Lambdapi
formula $\widehat{T}$. In the diagram, a certifying transformation is applied to
the initial task $T$ and produces the resulting tasks $L \defeq T_1; \dots; T_n$
and a certificate $c$, elaborated as $c_{kernel}$. Our tool then generates the
type $ty(T, L)$ which is the Lambdapi formula stating that $\widehat{T_1},
\widehat{T_1}, \dots,$ and $\widehat{T_n}$ implies $\widehat{T}$, type that we
call the \emph{application correctness type}.
Finally, we check that the application is correct by generating a proof
term~$t(c_{kernel})$ and asking Lambdapi to check that $t(c_{kernel})$ has
type~$ty(T, L)$.

This approach assumes that we trust Lambdapi, its type checker and the embedding
of proof tasks (paragraphs~\ref{embedding} and \ref{embedding:task}). However the
proof term generation is not contained in the trust base: the way the term is
obtained does not matter as long as it has the requested type. Additionally, we have
defined terms in Lambdapi for each certificate (paragraph~\ref{proofterm}),
including certificates from interpreted theories
(paragraph~\ref{proofterm:interp}). These terms have been checked by Lambdapi to
have the expected type so this gives us a machine-checked proof of
Theorem~\ref{ccorrect}.

\begin{theorem}
Consider a transformation application that from task~$T$ returns the list of
tasks~$L$. If the application correctness type~$ty(T, L)$ is
inhabited, then this application is correct.
\end{theorem}

\subsubsection{Shallow Embedding} \label{embedding}

In Lambdapi, we define the translation of a task validity by quantifying over
type symbols and function symbols, thus making these declarations explicit. We
are able to quantify in this way, both at the level of types and at the level of
terms, by using an encoding of the Calculus of
Constructions~\cite{DBLP:conf/fsttcs/Huet87} (written CoC) inside Lambdapi. We
obtain a formal description of the whole task which allows us to state and prove
the correctness of a transformation application.
For the system to stay coherent we should be careful when adding rewriting rules
and axioms (symbols in Lambdapi). We make use of an existing CoC encoding inside
Lambdapi~\cite{PTS} to which we add the axiom of excluded middle. This encoding
is also a shallow embedding inside Lambdapi, so we also get a shallow embedding
of our language inside Lambdapi. In this way, we do not need to explicitly
mention the context of proof and to handle it through inversion and weakening
lemmas, which would make the method impracticable.

We are able to translate our whole formalism using this embedding. The
translation of a term, a type or a task $t$ inside Lambdapi is
denoted~$\widehat{t}$.  We use exclusively the CoC syntax to describe this
translation: we write $\forall x : A, ~B$ for the dependent product, $A
\rightarrow B$ when $B$ does not depend on $x$, $\lambda x : A, ~B$ for the
abstraction and omit $A$ when it can easily be inferred. The sorts are $Type$
and $Kind$, with $Type$ being of type $Kind$. To translate the terms, we use an
impredicative encoding~\cite{PfenningP89}. Here is an excerpt of this encoding:

\[
\begin{array}{rllrl}
  \widehat{prop} ~\defeq& Type
  &\quad& \widehat{\bot} ~\defeq& \forall C : Type, ~C\\
  \widehat{t_1 \land t_2} ~\defeq& \forall C : Type, ~(\widehat{t_1} \rightarrow \widehat{t_2} \rightarrow C) \rightarrow C &&
  \widehat{\top} ~\defeq& \widehat{\bot} \rightarrow \widehat{\bot}
  \\
  \widehat{t_1 \lor t_2} ~\defeq& \forall C : Type, ~(\widehat{t_1} \rightarrow C) \rightarrow (\widehat{t_2} \rightarrow C) \rightarrow C &&
  \widehat{\lnot~ t} \defeq& \widehat{t} \rightarrow \widehat{\bot}
\end{array}
\]

We note $\widehat{\lnot} ~u$ for $u \rightarrow \widehat{\bot}$ such
that $\widehat{\lnot ~ t} = \widehat{\lnot} ~ \widehat{t}$ and we extend this
notation to the conjunction and the disjunction. Note that $\widehat{\top}$ is
inhabited by $\lambda c, ~c$.

\subsubsection{Translating Tasks}\label{embedding:task}

Let us give the translation of a task, where $Type^n$ denotes the $n$-ary
function over $Type$ (for example, $Type^2$ is
${Type \rightarrow Type \rightarrow Type}$).
\[
\begin{array}{lll}
  \text{Task } I ~|~ \Sigma ~|~ \Gamma \vdash \Delta \text{ with :}&& \text{Corresponding Lambdapi term:}\\
  I  \defeq \iota_1 : i_1, \dots, \iota_m : i_m & \quad \quad \quad& \forall \iota_1 : Type^{i_1}, ~\dots, \forall \iota_m : Type^{i_m}, \\
  \Sigma  \defeq f_1 : \tau_1, \dots, f_n : \tau_n && \quad \forall f_1 : \widehat{\tau_1},~ \dots, \forall f_n : \widehat{\tau_n}, \\
  \Gamma  \defeq H_1 : t_1, \dots, H_k : t_k && \qquad \widehat{t_1} \rightarrow \dots \rightarrow \widehat{t_k} \rightarrow \\
  \Delta  \defeq G_1 : u_1, \dots, G_l : u_l && \qquad\quad \widehat{\lnot~ u_1} \rightarrow \dots \rightarrow  \widehat{\lnot~ u_l} \rightarrow \widehat{\bot}
\end{array}
\]
Note that for polymorphic symbols, we need to declare them with extra type
parameters and to apply them to the appropriate type in the translation.

\subsubsection{Proof Term}\label{proofterm}

For each kernel rule, we associate a Lambdapi type and define a term that has this
type. When building a proof term, we first introduce the identifiers of the
resulting tasks and the identifiers of the type symbols, function symbols and
name of the premises of the initial task. Then, we translate the whole
certificate using the terms having the associated types. The $\Hole$ certificate
is special and does not have a fixed associated type. Instead, it is translated
as the identifier of the task it contains applied to its symbols and premises
following the same order that they have been introduced in. Assuming that our
encoding is correct, the fact that we define a term for every rule of the
certificate semantics gives us a machine-checked proof of the certificate
correctness, Theorem~\ref{ccorrect}.

To produce the proof term, we benefit from the elaboration of certificate in two
ways. First, the fact that the kernel certificates are elementary also
facilitates the definition of the terms corresponding to a kernel rule. Second,
each kernel certificate comes with additional data that we are able to use to
define such terms.

\begin{example} \label{ex:split2}
  For the $\Split$ rule presented in Figure~\ref{fig:crules}, we define a
  Lambdapi term $split$ that has the associated type
  $\forall t_1 : Type, ~\forall t_2 : Type, ~(t_1 \rightarrow \widehat{\bot})
  \rightarrow (t_2 \rightarrow \widehat{\bot}) \rightarrow t_1~\widehat{\lor}~ t_2
  \rightarrow \widehat{\bot}$.  We check the application of
  Example~\ref{ex:split} by verifying that the type
\begin{align*}
(&\forall x_1, ~\forall x, ~x_1 \rightarrow \widehat{\lnot} ~x \rightarrow \widehat{\bot}) \rightarrow\\
(&\forall x_2, ~\forall x, ~x_2 \rightarrow \widehat{\lnot} ~x \rightarrow \widehat{\bot}) \rightarrow\\
&\forall x_1, ~\forall x_2, ~\forall x, ~x_1 ~\widehat{\lor}~ x_2 \rightarrow \widehat{\lnot} ~x \rightarrow \widehat{\bot}
\end{align*}
is inhabited by the term
\begin{align*}
  &\lambda s_1, ~\lambda s_2, ~\lambda x_1, ~\lambda x_2, ~\lambda x, ~\lambda H, ~\lambda G,\\
  &split ~x_1 ~x_2~ (\lambda H, ~s_1 ~x_1 ~x ~H ~G)~ (\lambda H, ~s_2 ~x_2 ~x ~H ~G)
\end{align*}
Notice that $split$ takes the formulas it is applied to as arguments ($x_1$ and
$x_2$) and that those formulas have been found by elaborating the certificate.
\end{example}

\subsubsection{Encoding of Interpreted Theories} \label{proofterm:interp}

In Lambdapi, interpreted symbols are first declared in the preamble. When
interpreted symbols have corresponding certificate rules, we need to use the
properties of those symbols to prove that the types associated to these rules
are inhabited. Instead of declaring such symbols, we define them, which allows
us to prove the needed properties. Since our Lambdapi development is included in
the trusted code base of the Lambdapi checker, we make sure to only add axioms
and rewrite rules when necessary.

\paragraph{Polymorphic Equality.}

We define the equality in Lambdapi using the Leibniz definition of equality:
two terms $t_1$ and $t_2$ of type $\tau$ are equal when
$\forall Q : \tau \rightarrow Type, ~Q~ t_1 \rightarrow Q~ t_2$. Note that the
context of rewriting in the $\Rewrite$ rules is explicitly given as a
function. We use this function, translated as a Lambdapi function, to apply it
to the Leibniz equality when writing a proof term for a $\Rewrite$ certificate.

\paragraph{Integers.}

We define the integer type and usual integer operators in Lambdapi. First, we
define binary positive integers and binary negative integers. Integers are then
either $0$, positive or negative. We use rewrite rules to define those data
types which should be understood as algebraic data
types~\cite{CAC}. From these definitions, we get
a simple induction principle that we use to define a Lambdapi term corresponding
to the stronger induction principle described by the rule of the certificate
$\Induction$.

\section{Experimental Evaluation} \label{experiments}

The goal of this article is to provide a practical framework to render logical
transformations certifying. We apply the framework to Why3 and show that the
approach does not have an inherent problem of efficiency. More importantly, we
show that it is expressive enough to allow us to render complex transformations
certifying.

\subsection{Tests and Benchmarks} \label{tests}

We defined simple certifying transformations (about 15 of them) to
test every certificate.
We also defined a more complex transformation called \texttt{blast} meant to
discharge tautological propositional tasks. This transformation decomposes every
logical connector appearing at the head of formulas before calling itself
recursively. Rendering this transformation certifying required using the
composition of certifying transformations. We evaluated the efficiency of our
checkers by applying this transformation on problems of increasing size in
Figure~\ref{fig:benchs}. The problem with~$n$ propositional variables is to
verify that:
\[p_1 \Rightarrow (p_1 \Rightarrow p_2) \Rightarrow \dots (p_{n-1} \Rightarrow p_n) \Rightarrow p_n\]
We notice that the size of the kernel certificates is not linear with respect to
the number of variables. This is due to the fact that, contrary to the surface
certificates, the kernel certificates contain the formulas they are applied to.
Looking at the OCaml checker, our approach does not seem to have an inherent
problem of efficiency as the overhead it adds to the transformation is
negligible. On the other hand, the Lambdapi checker seems to be much slower.
Performances of Lambdapi have already been improved for our
purposes~\cite{p_shadow,p_context,p_nested,p_var_hyp,p_linear}, and we believe
they could be further improved in future versions. We could also modify our
checker to help Lambdapi, for example by abstracting away big formulas. We leave
this for future work.

\begin{figure}[tp]
  \centering
\begin{tabular}{||p{0.15\textwidth}||c|c|c|c|c|c|c|c|c|c||}
  \hline
  {\small number of variables} & $5$ & $10$ & $15$ & $20$ & $25$ & $50$ & $100$ & $200$ & $400$ & $800$\\
  \hline
  {\small transformation time (sec)} & $\sim 0$ & $\sim 0$ & $0.008$ & $0.016$ & $0.020$ & $0.080$ & $0.29$ & $1.21$ & $5.5$ & $25$ \\
  \hline
  {\small kernel certificate size (kB)} & $2.1$ & $5.8$ & $12$ & $19$ & $28$ & $85$ & $270$ & $950$ & $3500$ & $13000$\\
  \hline
  {\small OCaml checker time (sec)} & $\sim 0$ & $\sim 0$ & $\sim 0$ & $\sim 0$ & $\sim 0$ & $\sim 0$ & $\sim 0$ & $0.020$ & $0.084$ & $0.35$\\
  \hline
  {\small Lambdapi chec-ker time (sec)} & $0.072$ & $0.25$ & $0.80$ & $2.0$ & $4.2$ & $48$ & $660$ & - & - & -\\
  \hline
\end{tabular}
  \caption{Tests on Propositional Tasks}
  \label{fig:benchs}
  \hrulefill
\end{figure}

\subsection{Applications}

We evaluated our method by applying it at different levels. When rendering the
existing transformations \texttt{rewrite} and \texttt{induction} certifying, we
found that it is well-suited to add interpreted theories. When transformations
do not specifically deal with the theory in question, we do not need to extend
our certificate format, while, in general, the duo surface/kernel certificates
allows us to only define a minimal set of kernel rules, even if it means deriving
more surface certificates. By defining the Lambdapi checker, we gave a
machine-checked proof of the rules of our certificates which gives us confidence
in our certificates and their semantics.

\subsection{Better Understanding of Transformations} \label{destcase}

This work has led to a better understanding of transformations. On one hand,
instrumenting a transformation to produce an appropriate certificate requires to
understand why each application of this transformation is correct. Additionally,
once this is done, reading the certificate gives us another way to understand
why a particular transformation application is correct. On the other hand, this
work had led to the definition of the semantics of tasks inside Lambdapi and the
definition of the correction of a transformation in this setting.

In particular, type quantification is explicit in Section~\ref{terms}. For
example, the formula $\Pi \alpha.~ (\forall x : \alpha.~ \forall y : \alpha. ~x
= y) \lor \neg (\forall x : \alpha.~ \forall y : \alpha. ~x = y)$ means that
every type $\alpha$ either has at most one element or it has more than one. This
formula is provable but we cannot apply the certificate $\Split$ on such an
hypothesis. By contrast, in Why3, the type quantification is implicit, and it is
possible apply the \texttt{destruct} transformation on the hypothesis. This
gives us two resulting tasks: one with an hypothesis which states that every
type has at most one element, and the other with an hypothesis which states that
every type has more than one element, both being contradictory. This
bug~\cite{bugdestcase} has been found in the transformation \texttt{destruct}
when encoding proof tasks in Lambdapi; a similar bug was also found in the
transformation~\texttt{case}.

\section{Related Work}

To aid deductive program verification, a number of tools have been developed,
based on proof assistants or independently from them.
In the first case, the programming language on which the verification is done is
built from dedicated libraries and definition of both the programming
language and its logic inside the proof assistant. This is the case for example
for the library Iris~\cite{iris17} built on top of Coq and that allows reasoning
about concurrent, imperative programs or the library
\mbox{AutoCorres}~\cite{greenaway14pldi,greenaway15phd} built on top of Isabelle
and allowing to verify C programs. In such context, the correctness of the
approach is based on the formal semantics of programs and on deduction rules
established once and for all, which requires a large proof effort, thus limiting
the flexibility of the language.
In the second case, the tools developed are verifying annotated programs, and
generate proof obligations that are discharged by automatic theorem provers such
as SMT solvers. Examples of such tools are Why3, Dafny, Viper, Frama-C and SPARK.
Even though they rely on strong fundamental bases, their particular implementations
of such tools and some practical aspects such as their use of automatic theorem
provers have not been machine-checked and can contain bugs. An exception is
given by F*~\cite{DBLP:conf/popl/SwamyHKRDFBFSKZ16}, whose encoding's correctness to
SMT logic has been partially proved in Coq~\cite{Aguirre16}.

Our work lies in between these two approaches. On one hand logical
transformations are similar to tactics used in proof assistants such as
Coq~\cite{delahaye2000tactic}, except that our transformations are considered
part of the trusted code base. On the other hand, logical transformations can be
used automatically or interactively to help discharging proof obligations.  We
followed a skeptical approach extended with a preprocessing step (namely the
elaboration of certificates) similarly to~\cite{chihani2013checking},
except that our framework allows to check higher order proofs, and that the
focus is put on the ease of production of certificates. Indeed, we aim at making
it as easy as possible to render transformations certifying. We designed two
checkers: one based on the reflexive approach, known to be very
efficient~\cite{DBLP:conf/flops/GregoireTW06,armandkeller11} and the other one
based on a shallow embedding into the Lambdapi proof assistant. When using a
shallow embedding, the correctness of the verification relies on the considered
proof tool's correctness which makes its proof much
easier~\mbox{\cite{contejean08types,DBLP:journals/corr/CauderlierH15}}.

\section{Conclusion} \label{conclusion}

We presented a framework to validate logical transformations based on a
skeptical approach. When defining certificates, we put an emphasis on
\emph{modularity} by having certificates with holes and, with the notions of
surface and kernel certificates, \emph{ease of use} without compromising the
checker's verification. We combined all of these notions and applied them to
Why3 by implementing the certificate generation for various transformations and
the certificate verification with two checkers. The first checker was written in
OCaml and uses a computational approach which makes it very efficient while the
second checker is based on Lambdapi and gives us formal guarantees to its
correctness. We extended our work by adding the interpreted theories of the
integers and of the polymorphic equality. This allowed us to instrument more
complex and existing transformations to produce certificates, such as
\texttt{induction} and \texttt{rewrite}. Finally, we validated our method during
development and through tests and benchmarks.

\paragraph{Future Work.}

The current application of our method to Why3 could be improved at different
levels. The first idea is to instrument more transformations to produce
certificates, with polymorphism elimination~\cite{bobotPaskevich2011poly} and
algebraic data type elimination being important challenges. As the number of
certifying transformations increases, we also want to improve the efficiency of the
verification. To do so, we consider two factors: first, we want to compress
certificates on the fly when combining them; second, we want to improve the
efficiency of the Lambdapi checker by allowing to reuse the context of proof
that does not change. Additionally, we consider adding support for more
interpreted theories, while keeping the number of axioms and rewrite rules added
to Lambdapi to a minimum.

A long term goal is to increase trust in other parts of Why3. For example, we
could improve trust when calling automatic theorem
provers~\cite{armandkeller11,DBLP:conf/itp/BohmeW10} or improve trust in the
proof task generation which would require to formalize the semantics of the Why3
programming language~\cite{filliatre99preuve}.

Finally, our method is not specific to Why3 and can be applied, in general, to
certified logical encodings. In particular, existing (certifying) transformations
could be used for encoding a proof assistant's logic into an automatic theorem
prover's logic in order to benefit from both systems.

\paragraph{Acknowledgments.}
We are grateful to Alexandrina Korneva for the English proofreading and to
Claude Marché, Chantal Keller and Andrei Paskevich for their constant support
and their helpful suggestions.

\clearpage

\nocite{garchery:hal-02384946}

\bibliographystyle{eptcs}
\bibliography{biblio,abbrevs,demons,demons2,demons3,team,crossrefs}

\clearpage
\appendix

\section{Typing} \label{trules}

The predicate $\type{\Sigma}{t}{\tau}$ holds when $t$ has no free type variables
and is of type $\tau$ in signature $\Sigma$ and is formally defined by the
following rules:

\begin{mathpar}
  \inferrule{I, \iota : 0 ~|~ \type{\Sigma}{t[\alpha \mapsto \iota]}{prop} \\
  \iota \not\in I}
  {I ~|~ \type{\Sigma}{(\Pi \alpha. ~t)}{prop}}

  \inferrule{\tau \text{ is a subtype of } \Sigma(x) \\
  \tau \text{ has no type variables}}
  {\type{\Sigma}{x}{\tau}}

  \inferrule{~}
  {\type{\Sigma}{\top}{prop}}

  \inferrule{~}
  {\type{\Sigma}{\bot}{prop}}

  \inferrule{\type{\Sigma}{t}{prop}}
  {\type{\Sigma}{\neg t}{prop}}

  \inferrule{\type{\Sigma}{t_1}{prop}\\
  \type{\Sigma}{t_2}{prop}}
  {\type{\Sigma}{t_1 ~ op ~ t_2}{prop}}

  \inferrule{\type{\Sigma}{t_1}{\tau' \leadsto \tau}\\
  \type{\Sigma}{t_2}{\tau'}}
  {\type{\Sigma}{t_1 ~ t_2}{\tau}}

  \inferrule{\type{\Sigma, x : \tau}{t}{prop}\\
  \tau \text{ has no type variables}\\
  x \not\in \Sigma}
  {\type{\Sigma}{ (\forall x : \tau. ~ t)}{prop}}

  \inferrule{\type{\Sigma, x : \tau}{t}{prop}\\
  \tau \text{ has no type variables}\\
  x \not\in \Sigma}
  {\type{\Sigma}{ (\exists x : \tau. ~ t)}{prop}}

  \inferrule{\type{\Sigma, x : \tau'}{t}{\tau}\\
  \tau' \text{ has no type variables}\\
  x \not\in \Sigma}
  {\type{\Sigma}{ (\lambda x : \tau'. ~ t)}{\tau' \leadsto \tau}}
\end{mathpar}

\section{Certificate rules} \label{crules}

For each kernel certificate appearing in this article, we give its corresponding
rules. Theses rules are taken from the set of rules definining the predicate
$\bs{T}{c}$.

\begin{mathpar}
  \inferrule{~}
  {\bs{\Gamma\vdash \Delta}
    {\Hole(\Gamma \vdash \Delta)}}

  \inferrule{~}
  {\bs{\Gamma, H:\bot \vdash \Delta}
    {\Trivial(false, H)}}

  \inferrule{~}
  {\bs{\Gamma \vdash \Delta, G:\top}
    {\Trivial(true, G)}}

  \inferrule{\bs{\Sigma ~|~ \Gamma \vdash \Delta, P : t}{c_1} \\
    \bs{\Sigma ~|~ \Gamma, P : t \vdash \Delta}{c_2}\\
    \type{\Sigma}{t}{prop}}
  {\bs{\Sigma ~|~ \Gamma \vdash \Delta}
    {\Assert(P,t,c_1,c_2)}}

  \inferrule{\bs{\Gamma, H : t_1 \vdash \Delta}{c_1}\\
    \bs{\Gamma, H : t_2 \vdash \Delta}{c_2}}
  {\bs{\Gamma, H : t_1 \lor t_2 \vdash \Delta}
                  {\Split(false, t_1, t_2, H,c_1,c_2)}}

  \inferrule{\bs{\Gamma \vdash \Delta, G : t_1}{c_1}\\
    \bs{\Gamma \vdash \Delta, G : t_2}{c_2}}
  {\bs{\Gamma \vdash \Delta, G : t_1 \land t_2}
                  {\Split(true, t_1, t_2, G,c_1,c_2)}
                  }

  \inferrule{\bs{\Gamma, H_1 : t_1, H_2 : t_2 \vdash \Delta}{c}}
  {\bs{\Gamma, H : t_1 \land t_2 \vdash \Delta}
                 {\Destruct(false, t_1, t_2, H,H_1,H_2,c)}
                 }

  \inferrule{\bs{\Gamma \vdash \Delta, G_1 : t_1, G_2 : t_2}{c}}
  {\bs{\Gamma \vdash \Delta, G : t_1 \lor t_2}
                 {\Destruct(true, t_1, t_2, G,G_1,G_2,c)}
                 }
\end{mathpar}
\begin{mathpar}
  \inferrule{\bs{\Sigma, y : \tau ~|~ \Gamma, H : t [x \mapsto y] \vdash \Delta}{c}\\
    y \text{ is fresh w.r.t. } \Sigma, \Gamma, \Delta, t}
  {\bs{\Sigma ~|~ \Gamma, H : \exists x : \tau . ~t \vdash \Delta}
                  {\IntroQuant(false, \tau, \lambda x : \tau. ~t,H,y,c)}
                  }

  \inferrule{\bs{\Sigma, y : \tau ~|~ \Gamma\vdash \Delta, G : t [x \mapsto y]}{c}\\
    y \text{ is fresh w.r.t. } \Sigma, \Gamma, \Delta, t}
  {\bs{\Sigma ~|~ \Gamma\vdash \Delta, G : \forall x : \tau. ~t}
                  {\IntroQuant(true, \tau, \lambda x : \tau. ~t,G,y,c)}
                  }

  \inferrule{\bs{\Sigma ~|~ \Gamma, H_1 : \forall x : \tau. ~t, H_2 : t[x \mapsto u] \vdash \Delta}{c}\\
    \type{\Sigma}{u}{\tau}}
  {\bs{\Sigma ~|~ \Gamma, H_1 : \forall x : \tau. ~t \vdash \Delta}
                 {\InstQuant (false, \tau, \lambda x : \tau. ~t,H_1,H_2,u,c)}
                 }

  \inferrule{\bs{\Sigma ~|~ \Gamma \vdash \Delta, G_1 : \exists x : \tau. ~t, G_2 : t [x \mapsto u]}{c}\\
    \type{\Sigma}{u}{\tau}}
  {\bs{\Sigma ~|~ \Gamma \vdash \Delta, G_1 : \exists x : \tau. ~t}
                 {\InstQuant(true, \tau, \lambda x : \tau. ~t,G_1,G_2,u,c)}
                 }

  \inferrule{\bs{I, \iota : 0 ~|~ \Sigma ~|~ \Gamma\vdash \Delta, G : t [\alpha \mapsto \iota]}{c}\\
    \iota \not\in I}
  {\bs{I ~|~ \Sigma ~|~ \Gamma\vdash \Delta, G : \Pi \alpha.~ t}
                  {\IntroType(\Pi \alpha. ~t, G,\iota,c)}
                  }

  \inferrule{\bs{\Gamma, H_1 : \Pi \alpha. ~t, H_2 : t [\alpha \mapsto \tau] \vdash \Delta}{c}\\
    \tau \text{ has no type variables}}
  {\bs{\Gamma, H_1 : \Pi \alpha.~ t \vdash \Delta}
                 {\InstType (\Pi \alpha. ~t, H_1,H_2,\tau,c)}
                 }

  \inferrule{~}
  {\bs{\Sigma ~|~ \Gamma \vdash \Delta, G: x = x}
    {\EqRefl(x, G)}}

  \inferrule{\bs{\Gamma, H :  a = b, P : t[b] \vdash \Delta}{c}}
  {\bs{\Gamma, H : a = b, P : t[a] \vdash \Delta}
    {\Rewrite(false, a, b, t, P, H,c)}}

  \inferrule{\bs{\Gamma, H :  a = b \vdash \Delta, P : t[b]}{c}}
  {\bs{\Gamma, H : a = b \vdash \Delta, P : t[a]}
    {\Rewrite(true, a, b, t, P,H,c)}}

  \inferrule{i \text{ is fresh w.r.t. } \Gamma, \Delta, t \\
        \type{\Sigma}{i}{int}\\
        \type{\Sigma}{a}{int}\\
    \bs{\Gamma, H_i : i \leq a \vdash \Delta, G : t[i]}{c_{base}}\\
    \bs{\Gamma, H_i : i > a, H_{rec} : \forall n : int, ~ n < i \Rightarrow t[n] \vdash \Delta, G : t[i]}{c_{rec}}}
  {\bs{\Gamma \vdash \Delta, G : t[i]}{\Induction(i, a, t, G, H_i, H_{rec}, c_{base}, c_{rec})}}
\end{mathpar}

\end{document}